\documentclass[twocolumn]{article}

\usepackage{setspace}
\usepackage{amsmath}
\usepackage{amsfonts}
\usepackage{amsthm}
\usepackage{tikz}
\usepackage{bm,times}
\usepackage{amssymb}
\usepackage{graphicx}
\usepackage{multirow}
\usepackage{varwidth}
\usepackage{placeins}
\usepackage{blindtext}
\usepackage{enumitem}
\usepackage{colortbl}
\usepackage{breakcites}
\usepackage{natbib}
\usepackage[]{algorithm}
\usepackage{algorithmic}
\usepackage{graphics} 
\usepackage{epsfig} 
\usepackage{pgfplots}
\pgfplotsset{compat=1.18}
\usepackage{pgfplotstable}
\usepackage[latin1]{inputenc}
\usepackage[english]{babel}
\usepackage{graphicx}
\usepackage{textcomp}
\usepackage{hyperref}

\allowdisplaybreaks

\theoremstyle{definition} 
\theoremstyle{definition} 
\theoremstyle{definition} \newtheorem{theorem}{Theorem}
\theoremstyle{definition} 
\theoremstyle{definition}

\newtheorem{assump}{Assumption}

\newtheorem*{remark}{Remark}

\providecommand{\keywords}[1]
{
  \small  
  \textbf{\textit{Keywords---}} #1
}

\title{Steady State Distributed Kalman Filter}
\author{Francisco F. C. Rego$^{1}$
\\
        \small $^{1}$\textit{Universidade Lus\'{o}fona, INESC INOV - Lab, Lisboa, Portugal.} \\
}
\date{} 

\begin{document}
\maketitle

\begin{abstract}
This paper addresses the synthesis of an optimal fixed-gain distributed observer for discrete-time linear systems over wireless sensor networks. The proposed approach targets the steady-state estimation regime and computes fixed observer gains offline from the asymptotic error covariance of the global distributed BLUE estimator. Each node then runs a local observer that exchanges only state estimates with its neighbors, without propagating error covariances or performing online information fusion. Under collective observability and strong network connectivity, the resulting distributed observer achieves optimal asymptotic performance among fixed-gain schemes. In comparison with covariance intersection-based methods, the proposed design yields strictly lower steady-state estimation error covariance while requiring minimal communication. Numerical simulations illustrate the effectiveness of the approach and its advantages in terms of accuracy and implementation simplicity.
\end{abstract} \hspace{10pt}


\section{Introduction}

In many large-scale sensing applications, the state of a dynamical system must be estimated by a network of spatially distributed sensors connected through a communication graph. In such settings, no single sensor has access to sufficient information to reconstruct the global state, and communication constraints prevent the use of fully centralized architectures. Distributed state estimation aims at enabling each node to compute a local estimate of the global state by combining its own measurements with information received from neighboring nodes. Distributed estimation arises in applications such as network localization, environmental monitoring, and cooperative tracking under communication constraints (e.g., \cite{Akyildiz2002Survey,xu2002survey,Bahr2009Cooperative,Soares2013Joint} and references therein).

A wide range of distributed estimation algorithms has been proposed in the literature, particularly for linear time-varying and nonlinear systems. Among these, distributed Kalman filtering and covariance intersection (CI)-based methods are widely used due to their robustness to partial information and unknown correlations. Methods based on covariance intersection (CI) cope with unknown correlations by maintaining consistent upper bounds, at the cost of conservatism \cite{JulierUhlmann1997CI}. Alternative approaches based on distributed observer design have also been proposed for linear time-varying and communication-constrained settings, including constructibility Gramian-based methods and quantized observer schemes \cite{rego2023distributed,rego2016design}. However, these approaches typically rely on the online propagation and fusion of error covariances or information matrices, which results in significant computational and communication overhead. Importantly, this overhead persists even when the estimation process reaches steady state. Consensus-based distributed Kalman filtering schemes mitigate information imbalance through iterative consensus steps, but typically require additional communication and/or covariance-related computations \cite{OlfatiSaber2005DKF,Carli2007ConsensusDKF,BattistelliChisci2014KLA,Battistelli2015CBF,TalebiWerner2019EmbeddedConsensus}.

In many practical scenarios, the transient behavior of the estimator is of limited interest, while performance and simplicity in steady state are the primary concerns. This observation motivates a shift in perspective: rather than designing optimal time-varying distributed filters, one may seek to synthesize distributed observers with fixed gains that are optimal in the steady-state regime. Such observers can be implemented with substantially reduced information exchange, as they no longer require the online propagation of covariance information.

In this paper, we address the synthesis of an optimal fixed-gain distributed observer for discrete-time linear time-invariant systems. The proposed design exploits the asymptotic error covariance of the global distributed BLUE estimator to compute observer gains offline. During operation, each node runs a local observer and exchanges only state estimates with its neighbors. Under collective observability and strong connectivity assumptions, the resulting distributed observer achieves optimal asymptotic performance among fixed-gain schemes. In contrast to CI-based methods, the proposed approach yields strictly lower steady-state estimation error covariance while requiring minimal communication and computational effort.

\subsection{Notation}
Throughout this paper, we will use the symbol $\otimes$ for the Kronecker product. The notation $|\cdot|$ represents the cardinality of a set. $I_{M}$ denotes an $M\times M$ identity matrix, and $\boldsymbol{1}$ represents an $N\times 1$ vector with ones in every entry. When clear from the context, the superscript of a variable, e.g. $x^i$, refers to the node index of that variable, where $i\in\{1, \ldots , N\}:=\mathcal{N}$. The operator $\operatorname{row}(\cdot)$ is defined by $\operatorname{row}(X^i):=[X^1, \ldots , X^N]$, the operator $\operatorname{col}(\cdot)$ represents the column operator, i.e. $\operatorname{col}(X^i):=\operatorname{row}({X^i}^\intercal)^\intercal$, and the operator $\operatorname{diag}(X^i)$ yields a block diagonal matrix whose diagonal elements are $X^1, \ldots , X^N$. 

\section{Mathematical background}
\label{sec:math}

This section briefly recalls the elements of Best Linear Unbiased Estimation (BLUE) and Kalman filtering that are required for the synthesis of fixed-gain distributed observers. The presentation is intentionally concise and focuses on the role of error covariances in the computation of optimal observer gains.

\subsection{Best Linear Unbiased Estimation}

Consider the estimation of a deterministic variable $x \in \mathbb{R}^n$ from linear noisy measurements
\begin{equation*}
y = Fx + \epsilon,
\end{equation*}
where $\epsilon \sim \mathcal{N}(\boldsymbol{0},P)$ with $P \succ 0$, and $F \in \mathbb{R}^{m \times n}$ has full column rank. The Best Linear Unbiased Estimate (BLUE) of $x$ is given by
\begin{equation}
\label{eq:BLUE}
\hat{x} = (F^T P^{-1} F)^{-1} F^T P^{-1} y,
\end{equation}
and minimizes the estimation error covariance among all linear unbiased estimators \cite{verhaegen2007filtering,kailath2000linear,rao1973linear}. This result will be used repeatedly to characterize optimal gains in both centralized and distributed estimation settings.

\subsection{Kalman filtering and steady-state covariances}

Consider the discrete-time linear system
\begin{equation}
x_{t+1} = A x_t + w_t,
\end{equation}
with measurements
\begin{equation*}
y_t = C x_t + v_t,
\end{equation*}
where $w_t \sim \mathcal{N}(\boldsymbol{0},Q)$ and $v_t \sim \mathcal{N}(\boldsymbol{0},R)$, with $Q \succ 0$ and $R \succ 0$. The Kalman filter recursively computes a minimum-variance state estimate by alternating prediction and correction steps. The associated estimation error covariance $P_t$ evolves according to
\begin{equation}
P_t = \left( C^T R^{-1} C + \bar{P}_t^{-1} \right)^{-1}, \qquad
\bar{P}_{t+1} = A P_t A^T + Q.
\end{equation}

Under standard detectability assumptions on $(A,C)$, the covariance sequence converges to a constant matrix $P$, and the corresponding observer gains converge to fixed values. This steady-state property motivates the synthesis of fixed-gain observers, where optimal gains are computed offline from asymptotic error covariances and remain constant during operation. In the sequel, this principle is exploited to design distributed observers that achieve optimal steady-state performance while requiring reduced online information exchange.

\section{Problem Definition}
\label{sec:prob_def}

This section formulates the distributed estimation problem addressed in this paper, with emphasis on the steady-state regime and on the synthesis of fixed-gain distributed observers.

\subsection{Networked system}
\label{sec:sys_net}

We consider a discrete-time linear time-invariant system monitored by a network of sensing nodes. The setup consists of: (i) a discrete-time dynamical system; (ii) a set of nodes $\mathcal{N}$ with cardinality $N := |\mathcal{N}|$, each endowed with local sensing capabilities; and (iii) a communication network described by a directed graph $(\mathcal{N},\mathcal{A})$.

The system dynamics are given by
\begin{equation}
\label{eq:sys_dyn}
x_{t+1} = A x_t + w_t,
\end{equation}
where $x_t \in \mathbb{R}^n$ is the system state and $w_t \sim \mathcal{N}(\boldsymbol{0},Q)$ is the process noise. Each node $i \in \mathcal{N}$ acquires measurements of the form
\begin{equation}
\label{eq:meas_eq}
y_t^i = C^i x_t + v_t^i,
\end{equation}
where $v_t^i \sim \mathcal{N}(\boldsymbol{0},R^i)$ denotes the measurement noise.

\begin{assump}
\label{asmp:col_obs}
The pair $(A,C)$ is detectable, where $C := \operatorname{col}(C^i)$.
\end{assump}

Assumption~\ref{asmp:col_obs} requires only collective detectability and does not impose detectability of $(A,C^i)$ at each individual node.

\begin{assump}
\label{asmp:stochastic}
The process and measurement noises satisfy
\[
w_t \sim \mathcal{N}(\boldsymbol{0},Q), \qquad v_t^i \sim \mathcal{N}(\boldsymbol{0},R^i), \quad i \in \mathcal{N},
\]
with $Q \succ 0$ and $R^i \succ 0$, and are mutually uncorrelated across nodes and with the process noise.
\end{assump}

We further assume synchronized operation and ideal communication.

\begin{assump}
\label{asmp:synch}
All nodes are synchronized and operate at the same sampling rate.
\end{assump}

\begin{assump}
\label{asmp:finite_com}
Between two consecutive sampling instants, nodes can broadcast one message according to the communication graph $\mathcal{A}$.
\end{assump}

\subsection{Steady-state distributed estimation objective}

Under Assumptions~\ref{asmp:col_obs}--\ref{asmp:finite_com}, each node $i \in \mathcal{N}$ computes a local estimate $\hat{x}_t^i$ of the global state $x_t$ using its own measurements and information received from neighboring nodes.

In contrast with general distributed filtering problems, the focus of this paper is on the steady-state estimation regime. Specifically, we seek to synthesize distributed observers with constant gains that achieve optimal asymptotic estimation performance. The objective is to design observer gains such that, in steady state, the local estimation error covariances $P_t^i := E[e_t^i (e_t^i)^T]$ are minimized, subject to the information exchange constraints imposed by the network.

Throughout the paper, $\hat{x}_t^i$ denotes the state estimate at node $i$, $\bar{x}_t^i$ the corresponding predicted state, and $e_t^i := \hat{x}_t^i - x_t$ the estimation error. Global error and covariance quantities are defined by stacking local variables across nodes. These definitions allow us to characterize optimality in terms of the asymptotic error covariance of the global distributed BLUE estimator, which forms the basis for the fixed-gain observer synthesis developed in Section~\ref{sec:DKF}.

\section{Distributed observer synthesis based on BLUE}
\label{sec:DKF}

This section addresses the synthesis of distributed observers for linear systems with Gaussian process and measurement noise. Rather than proposing a generic distributed Kalman filtering algorithm, we use the classical BLUE formulation~\eqref{eq:BLUE} as a design tool to characterize optimal steady-state performance and to derive observer gains. We first recall a time-varying distributed BLUE formulation, which serves as a reference benchmark but requires the online propagation of global covariance information. We then show that, for linear time-invariant systems, the asymptotic error covariance of the global distributed BLUE estimator can be exploited to compute fixed observer gains offline. The resulting scheme is a steady-state distributed observer that exchanges only state estimates among neighboring nodes, achieves optimal asymptotic performance within the class of fixed-gain distributed observers, and significantly reduces communication and computational overhead when compared to consensus-based or covariance-intersection approaches.

\subsection{Time-varying distributed BLUE formulation}
\label{sec:KC}
We begin by recalling a distributed formulation of the BLUE estimator with time-varying gains, which serves as a reference for the development of the proposed fixed-gain observer. This formulation achieves optimal estimation performance in a stochastic sense, but requires the online propagation of global covariance and cross-covariance information, rendering it impractical for large-scale networks.

In detail, the method works as follows. Assume that we begin with an estimate $\bar{x}^i_t$ of the state $x_t$ with a Gaussian distribution and the following characteristics:
\begin{align}
\label{eq:KC_estimate}
E\left[\bar{x}^i_t-x_t\right]&=0, \\
E\left[\left(\bar{x}^i_t-x_t\right)\left(\bar{x}^j_t-x_t\right)^{T}\right]&:=\bar{P}^{ij}_t.
\end{align}
We define the global covariance matrix as $\bar{P}_t:=\left[\bar{P}^{ij}_t\right]_{ij\in\mathcal{N}}$ and the vector $\bar{x}_t:=\operatorname{col}\left(\bar{x}^i_t\right)$. The following theorem describes how to compute the BLUE of the state at the next time $\bar{x}^i_{t+1}$, given the local measurement $y^i_t$ and the estimates of the neighbours $\bar{x}^j_{t}$, $j\in\mathcal{N}^i$, as well as the global covariance matrix $P_{t+1}:=E\left[\bar{e}_{t+1}\bar{e}^{T}_{t+1}\right]$, where $\bar{e}_t:=\operatorname{col}(\bar{e}^i_t)$ with $\bar{e}^i_t:=\bar{x}^i_t-x_t$.

\begin{theorem}
\label{thm:KC}
Consider the matrices $\eta_i\in\mathbb{R}^{\left|\mathcal{N}^i\right|n\times Nn}$, defined by $\eta_i:=\operatorname{row}\left(\boldsymbol{e}_j,j\in\mathcal{N}^i\right)\otimes I_n$, where vector $\boldsymbol{e}_i$ is a column vector with all entries equal to $0$ except for entry $i$ which is $1$, $\boldsymbol{1}_i \in \mathbb{R}^{\left|\mathcal{N}^i\right|\times n}$ defined by $\boldsymbol{1}_i:=\boldsymbol{1}\otimes I_n$, and $\Gamma_{ij}\in\mathbb{R}^{\left|\mathcal{N}^i\right|n\times n}$ defined by $\Gamma_{ij}:=\eta_i\left(\boldsymbol{e}_j\otimes I_n\right)$.

Define $\tilde{\Omega}^i_t:=\boldsymbol{1}_i^T\left(\eta_i\bar{P}_t\eta_i^T\right)^{\dagger}\boldsymbol{1}_i$ and $\Omega^i_t:=\tilde{\Omega}^i_t+S^i$, where $\left(\eta_i\bar{P}_t\eta_i^T\right)^{\dagger}$ is the Moore-Penrose pseudo-inverse of $\eta_i\bar{P}_t\eta_i^T$\footnote{which is equivalent to $\left(\eta_i\bar{P}_t\eta_i^T\right)^{-1}$ if $\eta_i\bar{P}_t\eta_i^T$ is full rank}.

Given the estimates $\bar{x}^i_t$; $i\in\mathcal{N}$, satisfying \eqref{eq:KC_estimate} and the global covariance matrix $\bar{P}_t:=\left[\bar{P}^{ij}_t\right]_{ij\in\mathcal{N}}$ and defining the correction terms $s^{i}_t:=\left(C^i\right)^{T}V^iy^{i}_t$ and $S^i:=\left(C^i\right)^{T}V^iC^i$, the evolution of the BLUE of the state at time $t+1$, at node $i$, given the local measurement $y^i_t$ and the estimates of the neighbours $\bar{x}^j_{t}$, $j\in\mathcal{N}^i$, is described by
\begin{equation}
\bar{x}^{i}_{t+1}=A\left(\Omega^i_t\right)^{-1}\left(\sum_{j\in\mathcal{N}^i}\boldsymbol{1}_i^T\left(\eta_i\bar{P}_t\eta_i^T\right)^{\dagger}\Gamma_{ij}\bar{x}^j_t+s^i_t\right)
\end{equation}
whereas global covariance matrix is described by
\begin{equation}
\label{eq:P_KC_update}
\bar{P}_{t+1}=T_t\bar{P}_tT^{T}_t+\operatorname{diag}\left(AP^i_tS^{i}P^i_tA^{T}\right)+\left(\boldsymbol{1}_N\boldsymbol{1}_N^T\right)\otimes Q,
\end{equation}
where $T_t$ is defined as $T_t:=\left[T^{ij}_t\right]$, with
\begin{align*}
T^{ij}_t:=\left\{\begin{array}{lr}A P^i_t\boldsymbol{1}_i^T\left(\eta_i\bar{P}_t\eta_i^T\right)^{\dagger}\Gamma_{ij},& j\in \mathcal{N}^i \\
0,& j\notin \mathcal{N}^i \end{array}\right.
\end{align*}
\end{theorem}
\begin{proof}
See the Appendix.
\end{proof}
Although this formulation achieves optimal estimation performance, it requires each node to access global covariance and cross-covariance information, whose dimension grows with the network size. Moreover, the gains must be recomputed online and depend on the full covariance matrix $\bar{P}_t$. These requirements make the approach unsuitable for long-term operation and motivate the search for fixed-gain distributed observers.

\subsection{Fixed-gain steady-state distributed observer}
\label{sec:SSDKF}

We now depart from time-varying distributed filtering and address the synthesis of a fixed-gain distributed observer operating in steady state. The objective is no longer to recursively propagate error covariances or to adapt observer gains online, but rather to compute constant gains offline from asymptotic error statistics and to deploy them in a purely state-based distributed observer.

Specifically, we consider linear time-invariant systems for which the global covariance recursion associated with the time-varying distributed BLUE formulation in Section~\ref{sec:KC} converges to a stationary matrix $\bar{P}$. This asymptotic covariance characterizes the steady-state performance of the global distributed BLUE estimator and serves as the reference object for observer synthesis. Exploiting this property, we show that fixed observer gains can be computed offline and subsequently used by each node to perform distributed state estimation by exchanging only state estimates with neighboring nodes.

\subsubsection{Offline gain computation}

Assume that the global covariance recursion~\eqref{eq:P_KC_update} converges to a steady-state matrix $\bar{P} \succ 0$. In steady state, the information matrices $\Omega_t^i$ become constant and are given by
\[
\Omega^i := \boldsymbol{1}_i^T \left(\eta_i \bar{P} \eta_i^T\right)^{\dagger} \boldsymbol{1}_i + S^i,
\]
where $S^i := (C^i)^T V^i C^i$.

The fixed observer gains are then obtained by freezing the time-varying gains of Section~\ref{sec:KC} at their asymptotic values. For each node $i \in \mathcal{N}$ and each neighbor $j \in \mathcal{N}^i$, define
\[
D^{ij} := A (\Omega^i)^{-1}
           \boldsymbol{1}_i^T \left(\eta_i \bar{P} \eta_i^T\right)^{\dagger}
           \Gamma_{ij},
\]
and
\[
F^i := A (\Omega^i)^{-1} (C^i)^T V^i .
\]

These matrices are computed once, offline, using knowledge of the system model, noise statistics, and network topology. No covariance or cross-covariance information is required during online operation.

For clarity, the offline gain computation procedure is summarized in Algorithm~\ref{alg:offline_gains}.

\begin{algorithm}[H]
\caption{Offline computation of fixed distributed observer gains}
\label{alg:offline_gains}
\begin{algorithmic}[1]
\REQUIRE System matrices $A$, $\{C^i,R^i\}_{i\in\mathcal{N}}$, process noise covariance $Q$, network topology $\mathcal{A}$
\ENSURE Fixed gains $\{D^{ij},F^i\}$

\STATE Initialize $\bar{P}_0 \succ 0$
\REPEAT
    \STATE Update $\bar{P}_{k+1}$ using the global covariance recursion~\eqref{eq:P_KC_update}
\UNTIL{$\|\bar{P}_{k+1}-\bar{P}_k\| < \varepsilon$}
\STATE Set $\bar{P} \leftarrow \bar{P}_{k+1}$

\FOR{each node $i \in \mathcal{N}$}
    \STATE Compute $\Omega^i := \boldsymbol{1}_i^T(\eta_i \bar{P} \eta_i^T)^{\dagger}\boldsymbol{1}_i + S^i$
    \FOR{each neighbor $j \in \mathcal{N}^i$}
        \STATE Compute $D^{ij} := A (\Omega^i)^{-1}
        \boldsymbol{1}_i^T (\eta_i \bar{P} \eta_i^T)^{\dagger} \Gamma_{ij}$
    \ENDFOR
    \STATE Compute $F^i := A (\Omega^i)^{-1} (C^i)^T V^i$
\ENDFOR
\end{algorithmic}
\end{algorithm}

\subsubsection{Online distributed observer}

Once the fixed gains $\{D^{ij},F^i\}$ have been computed offline, each node runs a simple distributed observer. At each time step, node $i$ receives the current state estimates from its neighbors and combines them with its own measurement to update its estimate according to
\begin{equation}
\label{eq:fixed_gain_observer}
\hat{x}^i_{t+1}
=
\sum_{j \in \mathcal{N}^i} D^{ij} \hat{x}^j_t
+
F^i y^i_t .
\end{equation}

Importantly, the online implementation does not require the propagation or fusion of covariance matrices, nor any iterative consensus procedure. The only information exchanged between nodes consists of state estimates of dimension $n$.

The online execution at a generic node is summarized in Algorithm~\ref{alg:online_node}.

\begin{algorithm}[H]
\caption{Online fixed-gain distributed observer at node $i$}
\label{alg:online_node}
\begin{algorithmic}[1]
\REQUIRE Fixed gains $\{D^{ij},F^i\}$, neighbor set $\mathcal{N}^i$
\STATE Initialize $\hat{x}^i_0$

\FOR{each time step $t$}
    \STATE Receive $\hat{x}^j_t$ from neighbors $j \in \mathcal{N}^i$
    \STATE Acquire local measurement $y^i_t$
    \STATE Update estimate:
    \[
    \hat{x}^i_{t+1}
    =
    \sum_{j \in \mathcal{N}^i} D^{ij} \hat{x}^j_t
    +
    F^i y^i_t
    \]
\ENDFOR
\end{algorithmic}
\end{algorithm}

The proposed scheme constitutes a distributed observer with constant gains whose online implementation requires minimal communication and computation. The observer gains are computed offline from the asymptotic error covariance of the global distributed BLUE estimator and remain fixed during operation. As a result, the online execution requires only the exchange of state estimates among neighboring nodes, with no propagation or fusion of covariance information.

The design explicitly targets steady-state operation and provides a clear separation between offline observer synthesis and online state estimation. Within the class of fixed-gain distributed observers compatible with the network topology, the proposed approach achieves the same steady-state performance as the global distributed BLUE estimator whenever convergence is attained, while avoiding the conservatism and communication overhead typically associated with covariance-intersection and consensus-based filtering methods.

\begin{remark}[Convergence considerations]
\label{rem:convergence}

A formal proof of convergence of the global covariance recursion~\eqref{eq:P_KC_update} and of the resulting fixed-gain distributed observer is beyond the scope of this paper. Nevertheless, numerical experiments consistently indicate convergence of the covariance iteration and stable steady-state behavior of the proposed observer. These empirical observations suggest that the proposed fixed-gain design is robust in practice, even though a complete theoretical characterization of its convergence properties remains an open problem.
\end{remark}

\section{Numerical results}
\label{sec:example}

In this section, we illustrate the performance of the proposed distributed observer through numerical experiments. The objective is twofold: to empirically assess the steady-state behavior of the fixed-gain design and to compare its estimation performance with that of a consensus-based distributed Kalman filtering method~\cite{battistelli_consensus-based_2015}. In addition, we report empirical convergence properties of the global covariance iteration used for the offline gain computation.

\subsection*{Simulation setup}

We consider a distributed system of the form~\eqref{eq:sys_dyn}--\eqref{eq:meas_eq} that is collectively observable but not locally observable. The network consists of $N=20$ nodes. The system dynamics are defined by the identity matrix $A := I_n$, corresponding to decoupled agent dynamics.

Let $e^i$ denote the canonical row vector with a $1$ in position $i$ and zeros elsewhere. The local observation matrices are defined as
\begin{equation*}
C^i :=
\begin{bmatrix}
e^{i\intercal} - e^{(i+1)\intercal} \\
e^{(i-1)\intercal} - e^{i\intercal}
\end{bmatrix},
\end{equation*}
except for $i=1$, where $i-1$ is replaced by $N$, and for $i=N$, where $C^N := e^{N\intercal}$. This sensing configuration introduces coupling through the measurements while preserving decoupled dynamics, resulting in collective but not local observability and therefore requiring distributed state estimation.

Process and measurement noise are modeled as zero-mean Gaussian disturbances with covariances $Q = I_{2N}$ and $R^i = I_{m_i}$, respectively. The initial state is drawn from a Gaussian distribution with covariance $P_0 = 10^{6} I_{2N}$. The communication network is an undirected ring, where each node exchanges information with its two immediate neighbors.

\subsection*{Empirical convergence of the global covariance}

The fixed observer gains are computed offline from the steady-state solution of the global covariance recursion. Since a general convergence proof is not available, we assessed convergence empirically through a Monte Carlo study. Using the same baseline configuration defined previously, we performed $10$ independent runs with randomized system and measurement matrices. The covariance iteration was deemed converged when the Frobenius norm of the difference between successive iterates fell below $10^{-4}$.

Across all runs, convergence was consistently observed. The number of iterations required to reach convergence ranged from $36$ to $124$, with a mean of $73.3$ iterations and a median of $65.5$ iterations. These results indicate that, for the considered class of systems and network topologies, the offline covariance computation converges reliably within a moderate number of iterations, supporting the practical feasibility of the proposed fixed-gain design.

\subsection*{Estimation performance}

Figure~\ref{fig:results} compares the average estimation error norm
\[
\frac{1}{N} \sum_{i \in \mathcal{N}} \lVert \hat{x}^i_t - x_t \rVert
\]
obtained with different estimation strategies: a centralized Kalman filter, the consensus-based distributed Kalman filter of~\cite{battistelli_consensus-based_2015}, the proposed distributed estimator with time-varying gains, and its steady-state fixed-gain counterpart.

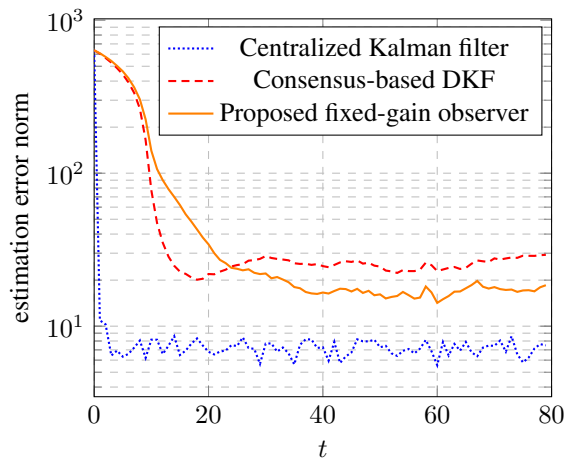
\begin{figure}[ht!]
\centering
\pgfplotstableread[header=false]{necent.dat}{\loadedtable}
\pgfplotstabletranspose{\necent}{\loadedtable}
\pgfplotstableread[header=false]{nedist.dat}{\loadedtable}
\pgfplotstabletranspose{\nedist}{\loadedtable}
\pgfplotstableread[header=false]{nedistsskf.dat}{\loadedtable}
\pgfplotstabletranspose{\nedistsskf}{\loadedtable}
\begin{tikzpicture}
\begin{axis}[
    ymode=log,
    xlabel={$t$},
    ylabel={estimation error norm},
    xmin=0, xmax=80,
    yminorgrids=true,
    xmajorgrids=true,
    grid style=dashed,
    width=0.95\linewidth,
]
    \addplot [blue,densely dotted,thick] table {\necent};
    \addplot [red,densely dashed,thick] table {\nedist};
    \addplot [orange,solid,thick] table {\nedistsskf};
    \legend{Centralized Kalman filter,Consensus-based DKF,Proposed fixed-gain observer}
\end{axis}
\end{tikzpicture}
\caption{Average norm of the estimation errors for different estimation strategies.}
\label{fig:results}
\end{figure}

The results show that the proposed fixed-gain distributed observer consistently achieves lower steady-state estimation error norms than the consensus-based distributed Kalman filter. Moreover, its performance closely matches that of the time-varying distributed formulation, while requiring substantially reduced online computation and communication. These observations provide empirical evidence that the proposed fixed-gain design attains the steady-state performance targeted by the offline covariance synthesis, in agreement with the convergence discussion presented earlier.

\section{Conclusion}
\label{sec:conclusion}

This paper addressed the synthesis of fixed-gain distributed observers for linear time-invariant systems under collective observability and limited communication. Using the BLUE formulation as a design tool, observer gains were computed offline from the steady-state solution of a global covariance recursion, leading to a distributed observer that exchanges only state estimates during online operation.

Numerical results demonstrated that the proposed observer exhibits reliable steady-state behavior and consistently outperforms a consensus-based distributed Kalman filtering approach in terms of estimation accuracy, while significantly reducing online computational and communication requirements. Although a general convergence proof for the covariance recursion remains an open problem, empirical evidence indicates that convergence is achieved in a moderate number of iterations for the considered scenarios. These results highlight the practical effectiveness of fixed-gain observer synthesis as a viable alternative to adaptive distributed filtering in LTI settings.

\section*{Acknowledgments}
The work of Francisco Rego was funded by national funds through FCT -- Fundação para a Ciência e a Tecnologia, I.P., under projects/supports UID/06486/2025 (\url{https://doi.org/10.54499/UID/06486/2025}), UID/PRR/06486/2025 (\url{https://doi.org/10.54499/UID/PRR/06486/2025}), and UID/PRR2/06486/2025 (\url{https://doi.org/10.54499/UID/PRR2/06486/2025}), and by COFAC/ILIND/COPELABS through the Seed Funding Program (7th edition), project LoRaMAR (\url{https://doi.org/10.62658/COFAC/ILIND/COPELABS/1/2025}).

\appendix
\section{Proof of Theorem \ref{thm:KC}}
\begin{proof}
After the agents communicate among themselves, the nodes compute $\tilde{x}^i_t$, the BLUE estimate given the estimates of the neighbours, i.e. given $\bar{x}^j_t,j\in \mathcal{N}^i$, where $\mathcal{N}^i$ is the set of neighbours of $i$. The estimate $\tilde{x}^i_t$ is the BLUE of $x_t$ such that $\eta_i\bar{x}_t=\boldsymbol{1}_ix_t+\epsilon_t$ with $\epsilon_t \sim \mathcal{N}\left(0,\eta_i\bar{P}_t\eta_i^T\right)$. Therefore from \eqref{eq:BLUE}, if $\bar{P}^{ii}_t$ is full rank, then
\begin{equation*}
\tilde{x}^i_t=\left(\bar{\Omega}^i_t\right)^{-1}\boldsymbol{1}_i^T\left(\eta_i\bar{P}_t\eta_i^T\right)^{\dagger}\eta_i\bar{x}_t
\end{equation*}
and $\tilde{x}^i_t$ has the following characteristics:
\begin{align*} 
E\left[\tilde{x}^i_t-x_t\right]&=0, \\
E\left[\left(\tilde{x}^i_t-x_t\right)\left(\tilde{x}^i_t-x_t\right)^T\right]&=\left(\tilde{\Omega}^i_t\right)^{-1},
\end{align*}
where $\tilde{\Omega}^i_t:=\boldsymbol{1}_i^T\left(\eta_i\bar{P}_t\eta_i^T\right)^{\dagger}\boldsymbol{1}_i$. The estimate $\tilde{x}^i_t$ can also be expressed as
\begin{equation*}
\tilde{x}^i_t=\left(\bar{\Omega}^i_t\right)^{-1}\boldsymbol{1}_i^T\left(\eta_i\bar{P}_t\eta_i^T\right)^{\dagger}\sum_{j\in\mathcal{N}^i}\Gamma_{ij}\bar{x}^j_t.
\end{equation*}
After taking a measurement, the nodes compute $\hat{x}^i_t$, the BLUE of $x_t$ given $\tilde{x}^i_t$ and $y^i_t$. This can be obtained directly as
\begin{equation*}
\hat{x}^i_t=\left(\Omega^i_t\right)^{-1}\left(\bar{\Omega}^i_t\tilde{x}^i_t+\left(C^{i}\right)^T\left(R^i\right)^{-1}y^i_t\right).
\end{equation*}
Finally, the nodes compute the prediction, i.e. the estimate of $x_{t+1}$ given $\hat{x}^i_t$, which is simply
\begin{equation*}
\bar{x}^{i}_{t+1}=A\hat{x}^i_t.
\end{equation*}
It now remains to compute the next global covariance matrix $\bar{P}_{t+1}$. For this purpose we analyze the error dynamics, i.e. the dynamics of $\bar{e}^i_t:=\bar{x}^i_t-x_t$. The dynamics of the estimate $\bar{x}^{i}_{t}$ can be written as follows:
\begin{align*}
\bar{x}^{i}_{t+1}&=A\left(\Omega^i_t\right)^{-1}\sum_{j\in\mathcal{N}^i}\boldsymbol{1}_i^T\left(\eta_i\bar{P}_t\eta_i^T\right)^{\dagger}\Gamma_{ij}\bar{x}^j_t\\
&+A\left(\Omega^i_t\right)^{-1}S^ix_t\\
&+A\left(\Omega^i_t\right)^{-1}\left(C^{i}\right)^T\left(R^i\right)^{-1}v^i_t.
\end{align*}
Given that $\bar{\Omega}^i_t=\sum_{j\in\mathcal{N}^i}\boldsymbol{1}_i^T\left(\eta_i\bar{P}_t\eta_i^T\right)^{\dagger}\Gamma_{ij}$ (from the fact that $\sum_{j\in\mathcal{N}^i}\Gamma_{ij}=\boldsymbol{1}_i$) the state dynamics can be expressed as
\begin{align*}
x_{t+1}&=Ax_t+w_t\\
&=A\left(\Omega^i_t\right)^{-1}\sum_{j\in\mathcal{N}^i}\boldsymbol{1}_i^T\left(\eta_i\bar{P}_t\eta_i^T\right)^{\dagger}\Gamma_{ij}x_t\\
&\quad+A\left(\Omega^i_t\right)^{-1}S^ix_t+w_t.
\end{align*}
The error dynamics can then be written as
\begin{align}
e^{i}_{t+1}&= A\left(\Omega^i_t\right)^{-1}\sum_{j\in\mathcal{N}^i}\boldsymbol{1}_i^T\left(\eta_i\bar{P}_t\eta_i^T\right)^{\dagger}\Gamma_{ij}e^j_t\nonumber\\
&+A\left(\Omega^i\right)^{-1}\left(C^{i}\right)^T\left(R^i\right)^{-1}v^i_t-w_t.
\end{align}
Defining $v_t=\operatorname{col}(v^i_t)$ yields $e_{t+1}=T_te_t+Kv_t+\boldsymbol{1}_N\otimes I_n w_t$. with $K:=\operatorname{diag}\left(A\left(\Omega^i_t\right)^{-1}\left(C^{i}\right)^T\left(R^i\right)^{-1}\right)$.
Finally, one computes the update law for $\bar{P}_t$ as in 
\eqref{eq:P_KC_update}.
\end{proof}

\bibliographystyle{plainnat}
\bibliography{tail/bibliography}
\end{document}